\documentclass{article}

\usepackage{amsmath}
\usepackage{amsfonts}
\usepackage{amsthm}

\newtheorem*{corollary}{Corollary}

\newtheorem*{unnumberedtheorem}{Theorem}

\begin{document}

\title{A Proof for Poisson Bracket in Non-commutative Algebra of Quantum Mechanics}

\author{Sina Khorasani\\
Sharif University of Technology, Tehran, Iran\\
Email: khorasani@sina.sharif.edu}

\maketitle

\begin{abstract}
The widely accepted approach to the foundation of quantum mechanics is that the Poisson bracket, governing the non-commutative algebra of operators, is taken as a postulate with no underlying physics. In this manuscript, it is shown that this postulation is in fact unnecessary and may be replaced by a few deeper concepts, which ultimately lead to the derivation of Poisson bracket. One would only need to use Fourier transform pairs and Kramers-Kronig identities in the complex domain. We present a definition of Hermitian time-operator and discuss some of its basic properties.

\end{abstract}

\section{Introduction} \label{Sec1}

The mathematical foundation of quantum mechanics is based on a few postulates, from which all the relationships and subsequent equations of quantum mechanics are obtained. A rigorous approach to this has been presented in famous textbooks by Sakurai \cite{Sakurai} and Schleich \cite{Schleich}, where operator algebra in Hilbert space, consisting of operators acting on kets and bras generates eigenket problems. The process of solving for Hermitian operator eigenkets and eigenvalues would constitute the whole mathematical toolbox of whatever needed to be known in quantum mechanics.

The idea of an algebraic operator of time is not new in the context of quantum mechanics. While in the early years physicist were trying to explore the implications of constructing a self-adjoint and Hermitian time operator, Wolfgang Pauli \cite{Pauli} was the first to point out that it was impossible to have such an operator of time. He did not present a well-founded proof, however, and his reason was that the time needed to be a continuous variable while energy eigenstates must be bounded from below \cite{Dirac}. This viewpoint has been followed in many of the later studies which have been discussed in several recent reviews of this subject \cite{Busch, PashbyThesis, Pashby}. It has been mostly believed since then that the time should be a real-valued parameter \cite{Aharanov, Durr, Cohen} and could not be expressed as a self-adjoint operator of any form. However, as it has been pointed out \cite{Grabowski} through presentation of a counter-example, this viewpoint is mathematically incorrect.

Despite all these reluctances to the existence of a time operator, lots of efforts have been put into deconstruction of this concept \cite{Yang, Garrison, Srinivas, Albanese, Olkhovsky}. Most recently, a mathematically consistent framework has been presented \cite{Pranovic} which adopts a self-adjoint and Hermitian operator of time. However, the ordinary Hamiltonian is not allowed to be a conjugate of time, again based on the assumption that time had to be a continuous variable while energy eigenstates should be bounded from the below. 

Last but not least, is that a discontinuous time variable with a discrete nature has a surprisingly interesting physical implication. As it has been shown by Golden \cite{Golden}, this would result in a fully-determinable dynamical theory of strictly-irreversible evolution, in which the conventional description of quantum-mechanical systems is maintained. Although a recent experiment \cite{Batalhao, Auffeves} has demonstrated time reversibility in precession of nuclear spin, it may be quite well because of the chaotic nature of the evolution of entangled systems which has been observed and reported \cite{Alidoosty1, Alidoosty2} elsewhere, in the sense that the final ending state of a quantum system is ultrasensitive to the initial conditions. Hence, the initial state may not be simply restored by preparation of initial conditions because of this chaotic nature. In that case, the true essence of time reversibility would be still preserved. These remain as open and unanswered questions which require further investigations in future.

\section{Definitions} \label{Sec2}

\subsection{Postulates} \label{Sec2.1}

\textit{Definition 1: Postulates of Quantum Mechanics.} All of the existing frameworks in which the quantum mechanics is founded are currently based on a few common posulates. These postulates may be enlisted as:

\begin{enumerate}
\item{\textit{Every physical state of a closed system may be uniquely represented by a normalized ket such as $|\psi\rangle$. Furthermore, system kets have one-to-one correspondence to their Hermitian adjoints, namely state bras such as $\langle\psi|$.}}
\item{\textit{To every physically measurable quantity, there corresponds exactly one Hermitian operator, which constitutes a Hilbert space with the space spanned by kets or bras.}}
\item{\textit{Squared absolute values of the inner product of state kets such as $|\psi\rangle$ with bras such as $\langle\phi|$, denoted by the squared bracket $|\langle\phi|\psi\rangle|^{2}$ gives a direct value of the probablity density function of the system projection on the state $|\phi\rangle$.}}
\item{\textit{The one-dimensional position and momentum operators along any direction do not commute, resulting in the Poisson bracket given by the expression $[\mathbb{X},\mathbb{P}]=i\hbar$.}}
\end{enumerate}

Out of these above mentioned postulates, with two more related to the measurement and collapse of state kets, the mathematical theory of quantum mechanics is born. Although the collapse of state kets may be well avoided by incorporation of multi-world theory \cite{Hall}, still an alternate postulate will be needed to justify the presence of parallel worlds.

As it appears in the above, the present refinement of the foundations of quantum mechanics requires the fourth postulate to suddenly come out of nowhere, with no physical justification, except if one is set to feed the corresponding Poisson bracket in classical mechanics from the behind scene into the quantum picture\cite{Pranovic}. It has been argued \cite{Messiah} that the commutation stated in the fourth postulate is a consequence of the operators associated with the corresponding observables, and thus a consequence of experiments. Although this relationship has been validated in practice on every experimental scale, and serves as the main basis of quantum mechanics, one would hardly accept its usefulness and correctness. It has been therefore an undeniable commonsense that there must be a more fundamental pillar alternative to this fourth postulate.

Recently, Bars and Rychkov \cite{Bars} have demonstrated that the Poisson bracket may be indeed obtained if we start with the mathematical behind dissociation and reconnection of strings in the much deeper universal picture of string theory. However, as it is being discussed in this paper, this level of complication is totally unnecessary, and the non-commutative algebra of position and momentum may be conveniently drawn from system evolution equation without need to worlds being in parallel or having tinier scales. To this end, all we need  to study the equivalence of wavefunctions will be the Fourier and the Hilbert transform pairs, the latter being used in the form of Kramers-Kronig relationships. It has to be further added that a connection between discrete Fourier transform and uncertainty relations for use in signal processing applications has been noticed in the past \cite{Massar}.
 
\subsection{Evolution of Physical Systems} \label{Sec2.2}

\textit{Definition 2: Evolution Equation \& Universal Parameter.} Upon projection unto a given $\zeta$ space (Appendix A), and in all non-relativistic and relativistic descriptions of physical systems, a governing differential equation appears in the general form

\begin{equation}
\label{eq1}
\mathcal{G}\psi \left( \zeta  \right) = \frac{\partial }
{{\partial \zeta }}\psi \left( \zeta  \right)
\end{equation}

\noindent
where $\mathcal{G}$ is an operator, $\psi \left( \zeta  \right)$ is the (scalar or vector) wavefunction, and $\zeta$ is a universal parameter, to which all other parameters and operators ultimately depend. Hence, the infinitesimal evolution of any system depending on the universal parameter $s$ is equivalent to operation of the operator $\mathcal{G}$ on its system wavefunction $\psi \left( \zeta  \right)$. Here, we do not discuss the nature of wavefunctions as these issues are still a matter of strong debate \cite{Hall,Wei}. However, the universal parameter $s$ may be regarded as a parameter such as time $t$, or any other physically measurable quantity such as the scalar quantities energy $e$, phase $\theta$, or vector quantities such as position $\mathbf{r}$ and momentum $\mathbf{p}$. With the exception of time $t$, which is shown to provide a truly a
ic behavior under interactions \cite{Barbour}, the choice of all other parameters would need a separate study. However, a single-particle system as being discussed in this paper, should experience symmetric time directions.\\

\noindent
\textit{Definition 3: Hermitian Quantum Evolution.} The evolution equation (\ref{eq1}) may be slightly modified to reach the Hermitian Quantum Evolution equation as  

\begin{equation}
\label{eq2}
\mathcal{G}(\zeta) \psi \left( \zeta  \right) = i\hbar\frac{\partial }
{{\partial \zeta }}\psi \left( \zeta  \right)
\end{equation}

\noindent
Here, the operator $\mathcal{G}$ is furthermore allowed to be a function of the universal parameter. Also, the unit imaginary number $i$ is inserted to preserve the Hermitian property of the derivative on the right-hand-side of (\ref{eq1}); this conforms to the well known principle of time invariance, too. Following the approach in \cite{Khorasani,Bars}, we also have inserted a dimensional constant which is denoted here by $\hbar$ with the units of Joule $\cdot$ Second to obtain the general form of the evolution equation matching the one generally used in quantum mechanics. 

The constant $\hbar$ will provide the necessary dimensional correction, in such a way that the dimension of the operator $\mathcal{G}$ multiplied by the dimension of parameter $\zeta$ would result in Joule $\cdot$ Second as well. The numerical value of the constant $\hbar$ in (\ref{eq2}) may be, however, later obtained by experiment, logically turning out to be the same as the measurable Planck's constant $h$ divided by $2\pi$. As it is known, the zero limit of this constant will reproduce the classical physics in the end.

Since the universal parameter $\zeta$ is supposed to be a truly independent variable, all other quantities may be assumed to vary as sole functions of $\zeta$. Hence, the dynamics of a physical system will be described by the variation of its quantities with respect to the variations of $\zeta$. At this point, we will have to take on the fourth postulate, replacing the one in the above as follows.\\

\noindent
\textit{Definition 4: Fourth Postulate.} This definition below may replace the fourth postulate based on the Possion bracket as in the above \S \ref{Sec2.1}.

\begin{itemize}
\renewcommand{\labelitemi}{$4.$}
\item{\textit{The system state functions always obey a form of the Evolution equation (\ref{eq2}).}}
\end{itemize}

When taken as the universal parameter, time in the classical physics is only a fourth dimension of the space-time, albeit shown to have a preferred direction \cite{Barbour}; this was discussed in  \S \ref{Sec1}. Even in the general and special theories of relativity, time coordinate remains to be a mere fourth dimension \cite{Eckstein, KhorasaniSPIE}. Hence, the classical systems may be or not be dependent on the time. Quantum mechanics, however, provides a different iterpretation of time. It exceptionally is an independent parameter, to which no Hermitian operator is assigned; time is only a scalar free parameter. This is while all other physical quantities (c.f. third postulate) are expressible as functions of this free parameter, namely time $t$.\\

\noindent
\textit{Definition 5: Schr\"{o}dinger Equation.} Upon taking $\zeta=t$ with the dimension of Second, we may denote $\mathcal{G}$ by the Hermitian energy operator $\mathcal{H}$, having the dimension of Joule, thus arriving at the Schr\"{o}dinger equation as

\begin{equation}
\label{eq3}
\mathcal{H}(t) \psi \left( t \right) = +i\hbar \frac{\partial }
{{\partial t}}\psi \left( t \right)
\end{equation}

\noindent
\textit{Definition 6: Position and Momentum Operators.} Had we taken a vector parameter such as three-dimensional momentum $\mathbf{p}$ with the dimension of Kilogram $\cdot$ Meter $\cdot$ Second$^{-1}$ instead of the universal parameter $\zeta$, we had arrived at the vector equation

\begin{equation}
\label{eq4}
\mathcal{R}\psi \left( {\mathbf{p}} \right) =  + i\hbar \frac{\partial }
{{\partial {\mathbf{p}}}}\psi \left( {\mathbf{p}} \right)
\end{equation}

\noindent
Similarly, one would obtain the following upon taking the three-dimensional position $\mathbf{r}$ with the dimension of Meter instead of the universal parameter $\zeta$

\begin{equation}
\label{eq5}
\mathcal{P}\chi \left( {\mathbf{r}} \right) =  - i\hbar \frac{\partial }
{{\partial {\mathbf{r}}}}\chi \left( {\mathbf{r}} \right)
\end{equation}

\noindent
In (\ref{eq4}) and (\ref{eq5}), $\mathcal{R}$ and $\mathcal{P}$ are respectively the Hermitian vector operators corresponding to position and momentum in space.

Although, $\chi(\mathbf{p})$ in (\ref{eq4}) and $\psi(\mathbf{r})$ in (\ref{eq5}) express identical systems, they are not necessarily equal, too. That would mean in general that $\chi(\mathbf{r})\neq\psi(\mathbf{r})$ and $\chi(\mathbf{p})\neq\psi(\mathbf{p})$. We refer to $\chi(\mathbf{p})$ and $\psi(\mathbf{r})$ respectively as the momentum and position representations of the system. The reason for taking the negative sign on the right-hand-side of (\ref{eq5}) will become apparent later. \\

\noindent
\textit{Definition 7: Fourier Transforms.} Assuming satisfaction of sufficient condition for existence of Fourier transform (absolute integrability), we may define the three-dimensional Fourier transform

\begin{eqnarray}
\label{eq6}
\Psi \left( {\mathbf{r}} \right)&& = \mathcal{F}\left\{ {\psi \left( {\mathbf{p}} \right)} \right\}\left( {\mathbf{r}} \right)  \\
\nonumber
&&
=\frac{1}
{{\left( {2\pi \hbar } \right)^{\tfrac{3}
{2}} }}\iiint {\psi \left( {\mathbf{p}} \right)\exp \left( { + \frac{i}
{\hbar }{\mathbf{r}} \cdot {\mathbf{p}}} \right)d^3 p}
\end{eqnarray}

\noindent
having the inverse transform given by

\begin{eqnarray}
\label{eq7}
\psi \left( {\mathbf{p}} \right) && = \mathcal{F}^{ - 1} \left\{ {\Psi \left( {\mathbf{r}} \right)} \right\}\left( {\mathbf{p}} \right) \\
\nonumber
&&= \frac{1}
{{\left( {2\pi \hbar } \right)^{\tfrac{3}
{2}} }}\iiint {\Psi \left( {\mathbf{r}} \right)\exp \left( { - \frac{i}
{\hbar }{\mathbf{p}} \cdot {\mathbf{r}}} \right)d^3 r}
\end{eqnarray}

\noindent
Now, opon taking Fourier transform from both sides of (\ref{eq4}) we arrive at

\begin{eqnarray}
\label{eq8}
\mathcal{F}&&\left\{ \mathcal{R}\psi \left( \mathbf{p} \right) \right\} \left( \mathbf{r} \right) \\
\nonumber
&&
= \frac{i\hbar }
{\left( {2\pi \hbar } \right)^{\tfrac{3}{2}} } \iiint \left[ \frac{\partial }
{\partial \mathbf{p}}\psi \left( \mathbf{p} \right) \right] \exp \left(  + \frac{i}{\hbar}
\mathbf{r} \cdot \mathbf{p} \right)d^3 p
\end{eqnarray}

\noindent
which through part-by-part integration gives

\begin{eqnarray}
\label{eq9}
\nonumber
\mathcal{F}&&\left\{ {\mathcal{R}\psi \left( {\mathbf{p}} \right)} \right\}\left( {\mathbf{r}} \right) = \frac{{i\hbar }}
{{\left( {2\pi \hbar } \right)^{\tfrac{3}
{2}} }}\left. {\psi \left( {\mathbf{p}} \right)\exp \left( { + \frac{i}
{\hbar }{\mathbf{r}} \cdot {\mathbf{p}}} \right)} \right|_{p =  - \infty }^{p =  + \infty }
\\ \nonumber &&
 - \frac{{i\hbar }}
{{\left( {2\pi \hbar } \right)^{\tfrac{3}
{2}} }}\iiint {\psi \left( {\mathbf{p}} \right)\left[ {\frac{\partial }
{{\partial {\mathbf{p}}}}\exp \left( { + \frac{i}
{\hbar }{\mathbf{r}} \cdot {\mathbf{p}}} \right)} \right]d^3 p}
\\ \nonumber &&
= \frac{{\mathbf{r}}}
{{\left( {2\pi \hbar } \right)^{\tfrac{3}
{2}} }}\iiint {\psi \left( {\mathbf{p}} \right)\exp \left( { + \frac{i}
{\hbar }{\mathbf{r}} \cdot {\mathbf{p}}} \right)d^3 p}
\\  &&
= {\mathbf{r}}\mathcal{F}\left\{ {\psi \left( {\mathbf{p}} \right)} \right\}\left( {\mathbf{r}} \right)
\end{eqnarray}

Now, to be discussed below, it is expected that (\ref{eq4}) and (\ref{eq5}) would describe identical systems. What we therefore require to know is that

\begin{equation}
\label{eq10}
\Psi \left( {\mathbf{r}} \right) = \mathcal{F}\left\{ {\psi \left( {\mathbf{p}} \right)} \right\}\left( {\mathbf{r}} \right) = \chi \left( {\mathbf{r}} \right)
\end{equation}

\noindent
Similarly, we would expect the inverse relationship to hold true

\begin{eqnarray}
\label{eq11}
\nonumber
\mathcal{F}^{ - 1}&& \left\{ {\mathcal{P}\chi \left( {\mathbf{r}} \right)} \right\}\left( {\mathbf{p}} \right) = \frac{{ - i\hbar }}
{{\left( {2\pi \hbar } \right)^{\tfrac{3}
{2}} }}\left. {\chi \left( {\mathbf{r}} \right)\exp \left( { - \frac{i}
{\hbar }{\mathbf{p}} \cdot {\mathbf{r}}} \right)} \right|_{r =  - \infty }^{r =  + \infty }   \\
\nonumber
&& + \frac{{i\hbar }}
{{\left( {2\pi \hbar } \right)^{\tfrac{3}
{2}} }}\iiint {\chi \left( {\mathbf{r}} \right)\left[ {\frac{\partial }
{{\partial {\mathbf{r}}}}\exp \left( { - \frac{i}
{\hbar }{\mathbf{p}} \cdot {\mathbf{r}}} \right)} \right]d^3 p}  \\
\nonumber
&& = \frac{{\mathbf{p}}}
{{\left( {2\pi \hbar } \right)^{\tfrac{3}
{2}} }}\iiint {\chi \left( {\mathbf{r}} \right)\exp \left( { - \frac{i}
{\hbar }{\mathbf{p}} \cdot {\mathbf{r}}} \right)d^3 p}  \\
&& = {\mathbf{p}}\mathcal{F}^{ - {\text{1}}} \left\{ {\chi \left( {\mathbf{r}} \right)} \right\}\left( {\mathbf{p}} \right)
\end{eqnarray}

\noindent
and therefore

\begin{equation}
\label{eq12}
{\rm X}\left( {\mathbf{p}} \right) = \mathcal{F}^{ - 1} \left\{ {\chi \left( {\mathbf{r}} \right)} \right\}\left( {\mathbf{p}} \right) = \psi \left( {\mathbf{p}} \right)
\end{equation}

Of principal importance, now, is to know whether $\Psi(\mathbf{r})=\chi(\mathbf{r})$ and ${\rm X}(\mathbf{p})=\psi(\mathbf{p})$ hold necessarily true or not. Counter-intuitively, it is not obvious that these two equations actually are correct. Similarly, we may assign an operator to time such as $\mathcal{T}$ and reconfigure the Schr\"{o}dinger equation (\ref{eq3}) in the alternate energy representation form as

\begin{equation}
\label{eq13}
\mathcal{T}(e) \chi \left( e \right) = -i\hbar \frac{\partial }
{{\partial e}}\chi \left( e \right)
\end{equation}

\noindent
\textit{Definition 8: Conjugate Variables.} Following (\ref{eq4}) and (\ref{eq5}), the system vector parameters $\mathbf{p}$ and $\mathbf{r}$ are here defined to be \textit{conjugate variables}. Similarly, following (\ref{eq13}) $e$ and $t$ now obviously again form a pair of \textit{conjugate variables}. Therefore, $\mathbf{p}$ and $\mathbf{r}$ as well as $\psi(t)$ and $\chi(e)$ should be connected through Fourier transforms.

We furthermore notice the negative sign on the right-hand-side of (\ref{eq13}) in comparison to (\ref{eq3}). 

\subsection{Fourth Postulate} \label{Sec2.3}

What allows us at this stage to proceed is, in fact, the assumption of a new amendment to the fourth postulate as follows.\\

\noindent
\textit{Definition 9: Amendment to the Fourth Postulate.} The fourth postulate of quantum mechanics links the evolution equation to the existence of conjugate pairs of physical observables.

\begin{itemize}
\renewcommand{\labelitemi}{$4.$}
\item{\textit{Physical systems may be described by either of the mathematically equivalent conjugate forms of the Evolution equation (\ref{eq2}), related through Fourier transform pairs. Conjugate forms should ultimately result in identical probablity density functions}}
\end{itemize}

\noindent
Having that said in the above, we must have at least $|\Psi(\mathbf{r})|=|\chi(\mathbf{r})|$ and $|{\rm X}(\mathbf{p})=\psi(\mathbf{p})|$. Equivalence of phases, comes following the Kramers-Kronig relationships. To verify this, it is sufficient that we assume that state functions are analytic in complex domain. Then, their natural logarithms given by $\ln\Psi(\mathbf{r})=|\Psi(\mathbf{r})|+i\angle\Psi(\mathbf{r})+i 2\pi n$ and  $\ln\chi(\mathbf{r})=|\chi(\mathbf{r})|+i\angle\chi(\mathbf{r})+i 2\pi m$ should also be analytic, where $m,n\in\textsf{Z}$ are some arbitrary integers. Kramers-Kronig relations \cite{Arfken} require that

\begin{eqnarray}
\label{eq14}
\Im{ f(z)} =  - \frac{1}{\pi }
\textrm{P.V.}\int\limits_{ - \infty }^{ + \infty } \frac{\Re f(u)}{u - z}du
\\ \nonumber
\Re{ f(z)} =  + \frac{1}{\pi }
\textrm{P.V.}\int\limits_{ - \infty }^{ + \infty } \frac{\Im f(u)}{u - z}du
\end{eqnarray}

\noindent
where $\textrm{P.V.}$ denotes Cauchy principal values. The first of the above pair is sufficient to observe that $\angle\Psi(\mathbf{r})-\angle\chi(\mathbf{r})=2\pi q$ with $q\in\textsf{Z}$ being some integer. Regardless of the choice of $q$, we would readily have (\ref{eq10}) established. Similarly, (\ref{eq12}) follows in the momentum representation.

\section{Poisson Bracket} \label{Sec3}

\textit{Definition 10: Commutator.} In order to study the commutative properties of any pair of two physical quantities, we may define the commutator of two operators such $\mathcal{A}$ and $\mathcal{B}$ as

\begin{equation}
\label{eq15}
[\mathcal{A},\mathcal{B}]\equiv\mathcal{A}\mathcal{B}-\mathcal{B}\mathcal{A}
\end{equation}

\noindent
A zero commutator implies insignificance in the order of measurements, or compatibility of observations. On the contrary, the order of measurements of two quantities will become important when the commutator is non-zero. \\

\noindent
\textit{Definition 11. Piosson Bracket.} For the particular choice of vector operators for position $\mathcal{R}$ and momentum $\mathcal{P}$, their commutator $[\mathcal{R},\mathcal{P}]$ will be referred to as the \textit{Poisson Bracket}. We may notice here that the Poisson bracket when evaluated on vector position $\mathcal{R}$ and momentum $\mathcal{P}$ operators gives out actually a tensor product value.

We assume that the individual components of position and momentum do commute, that is

\begin{eqnarray}
\label{eq16}
[\mathcal{R}_m,\mathcal{R}_n]=0
\\
\nonumber
[\mathcal{P}_m,\mathcal{P}_n]=0
\end{eqnarray}

\noindent
where $m,n=x,y,z$. It should be emphasized that these assumptions (\ref{eq16}) correspond to the case of simple vacuum in the absence of magnetic field and any spatial lattice, which admit vanishing commutators between the same space and momentum coordinates. Otherwise, one may need to take $[\mathcal{R}_m,\mathcal{R}_n]=i\hbar\delta_{mn}\eta$ and $[\mathcal{P}_m,\mathcal{P}_n]=i\hbar\delta_{mn}\beta$ where $\eta$ and $\beta$ are independent real-valued constants \cite{Castellani, Hatzinikitas, Delduc} proportional to the sqaure Planck length and magnetic field, respectively; such discussion is beyond the scope of this paper.

\begin{unnumberedtheorem}
The Possion bracket in quantum mechanics equals to the unit imagniary number $i$ times the reduced Planck constant. 
\end{unnumberedtheorem}

\begin{proof}
Now, to evaluate the value of Poisson bracket $[\mathcal{R},\mathcal{P}]$ we verify its effect on an arbitrary function of position such as $f(\mathbf{r})$

\begin{equation}
\label{eq17}
[\mathcal{R},\mathcal{P}]f(\mathbf{r})=(\mathcal{R}\mathcal{P}-\mathcal{P},\mathcal{R})f(\mathbf{r})=
\mathcal{R}\mathcal{P}f(\mathbf{r})-\mathcal{P}\mathcal{R}f(\mathbf{r})
\end{equation}

\noindent
Using (\ref{eq4}) and (\ref{eq5}) we obtain

\begin{eqnarray}
\label{eq18}
\nonumber
\left[ \mathcal{R},\mathcal{P} \right]f\left( {\mathbf{r}} \right) && =  \mathcal{R}\left( { - i\hbar \frac{\partial }
{{\partial {\mathbf{r}}}}} \right)f\left( {\mathbf{r}} \right) - \left( { - i\hbar \frac{\partial }
{{\partial {\mathbf{r}}}}} \right){\mathbf{r}}f\left( {\mathbf{r}} \right)
\\ \nonumber &&
=  - i\hbar  \mathcal{R}\frac{{\partial f\left( {\mathbf{r}} \right)}}
{{\partial {\mathbf{r}}}} + i\hbar {\rm I} f\left( {\mathbf{r}} \right) + i\hbar {\mathbf{r}}\frac{{\partial f\left( {\mathbf{r}} \right)}}
{{\partial {\mathbf{r}}}}
\\ \nonumber &&
=  - i\hbar {\mathbf{r}}\frac{{\partial f\left( {\mathbf{r}} \right)}}
{{\partial {\mathbf{r}}}} + i\hbar {\rm I} f\left( {\mathbf{r}} \right) + i\hbar {\mathbf{r}}\frac{{\partial f\left( {\mathbf{r}} \right)}}
{{\partial {\mathbf{r}}}}
\\  &&
= i\hbar {\rm I} f\left( {\mathbf{r}} \right)
\end{eqnarray}

\noindent
Here, ${\rm I}=\nabla \mathbf{r}$ is the $3\times 3$ identity matrix. In the derivation of the above, we have assumed the following vector identities for the position $\mathcal{R}$ and momentum $\mathcal{P}$ operators as

\begin{eqnarray}
\label{eq19}
\mathcal{R}f(\mathbf{r})=\mathbf{r} f(\mathbf{r})
\\ \nonumber
\mathcal{P}g(\mathbf{p})=\mathbf{p} g(\mathbf{p})
\end{eqnarray}

\noindent
with $f(\cdot)$ and $g(\cdot)$ being arbitrary functions. As a matter of fact, any function of position such as $f(\mathbf{r})$ must be an eigenfunction of $\mathcal{R}$ with the vector eigenvalue $\mathbf{r}$. Similarly, any function of momentm such as $g(\mathbf{p})$ must be an eigenfunction of $\mathcal{P}$ with the vector eigenvalue $\mathbf{p}$. Now, reverting back to (\ref{eq18}) since $f(\mathbf{r})$ was assumed to be completely arbitrary, we obtain the Poisson bracket as

\begin{equation}
\label{eq20}
[\mathcal{R},\mathcal{P}]=i\hbar {\rm I}
\end{equation}

\noindent
This completes the proof.
\end{proof}

\begin{corollary}
Had we assumed an arbitrary function of momentum in the form $g(\mathbf{p})$, we would again obtain the Poisson bracket (\ref{eq20}) as $\left[ {\mathcal{R},\mathcal{P}} \right]g\left( {\mathbf{p}} \right)=i\hbar {\rm I} g\left( {\mathbf{p}} \right)$.
\end{corollary}

\begin{proof}
Proof is easily obtained by direct expansion of the operators as

\begin{eqnarray}
\label{eq21}
\nonumber
\left[ {\mathcal{R},\mathcal{P}} \right]g\left( {\mathbf{p}} \right) &&= \left( { + i\hbar \frac{\partial }
{{\partial {\mathbf{p}}}}} \right){\mathbf{p}}g\left( {\mathbf{p}} \right) - \mathcal{P}\left( { + i\hbar \frac{\partial }
{{\partial {\mathbf{p}}}}} \right)g\left( {\mathbf{p}} \right)
\\ \nonumber &&
=  + i\hbar {\mathbf{p}}\frac{{\partial g\left( {\mathbf{p}} \right)}}
{{\partial {\mathbf{p}}}} + i\hbar {\rm I} g\left( {\mathbf{p}} \right) - i\hbar \mathcal{P}\frac{{\partial g\left( {\mathbf{p}} \right)}}
{{\partial {\mathbf{p}}}}
\\ \nonumber &&
=  + i\hbar {\mathbf{p}}\frac{{\partial g\left( {\mathbf{p}} \right)}}
{{\partial {\mathbf{p}}}} + i\hbar {\rm I} g\left( {\mathbf{p}} \right) - i\hbar {\mathbf{p}}\frac{{\partial g\left( {\mathbf{p}} \right)}}
{{\partial {\mathbf{p}}}}
\\  &&
= i\hbar {\rm I} g\left( {\mathbf{p}} \right)
\end{eqnarray}

\noindent
This will result in the same relationship as (\ref{eq20}). 
\end{proof}

It is customary to write the tensor Poisson bracket as

\begin{equation}
\label{eq22}
[\mathcal{R}_m,\mathcal{P}_n]=i\hbar \delta_{mn}
\end{equation}

\noindent
with $\delta_{mn}$ being Kronecker delta.

\subsection{Generalization} \label{Sec3.1}
Treatment of the more general choice of $f(\mathbf{r},\mathbf{p})$ is slightly more difficult, but will again result in the same conclusion of (\ref{eq20}). For this purpose, we first define the Weyl symmetrization operator.\\

\textit{Definition 12. Weyl Symmetrization Operator.} The Weyl symmetrization operator \cite{Schleich, Castellani, Hatzinikitas} denoted by $\textsf{S}$ is recursively defined, operating as

\begin{eqnarray}
\label{eq25}
\textsf{S}\{\mathcal{AB}\}&&=\frac{1}{2}(\mathcal{AB}+\mathcal{BA})
\\ \nonumber
\textsf{S}\{\mathcal{ABC}\}&&=\frac{1}{3}
(\mathcal{A}\textsf{S} \{\mathcal{BC}\}+\mathcal{B}\textsf{S}\{\mathcal{CA}\}+\mathcal{C}\textsf{S} \{\mathcal{BA}\})
\end{eqnarray}

\noindent
and so on.

\begin{unnumberedtheorem}
The expression for Poisson bracket (\ref{eq20}) holds true for the case of general choice of $f(\mathbf{r},\mathbf{p})$.
\end{unnumberedtheorem}

\begin{proof}
Generalization is obtained by expanding $f(\mathbf{r},\mathbf{p})$ first as

\begin{equation}
\label{eq23}
f(\mathbf{r},\mathbf{p}) = \sum\limits_{m,n = 0}^\infty  \frac{\mathbf{r}^n \mathbf{p}^m }
{n!m!} \cdot \frac{\partial ^{n + m} f(\mathbf{r},\mathbf{p})}
{\partial \mathbf{r}^n \partial \mathbf{p}^m } |_{\mathbf{r},\mathbf{p} = \mathbf{0}}
\end{equation}

\noindent
which allows us to appropriately define the operator

\begin{equation}
\label{eq24}
f(\mathcal{R},\mathcal{P}) = \sum\limits_{m,n = 0}^\infty  \frac{\partial ^{n + m} f(\mathbf{r},\mathbf{p})}
{\partial \mathbf{p}^n \partial \mathbf{p}^m } |_{\mathbf{r},\mathbf{p} = \mathbf{0}} \cdot
\frac{\textsf{S}\{\mathcal{R}^n \mathcal{P}^m \}}{n!m!}
\end{equation}

\noindent
It is now possible to use (\ref{eq18}) and (\ref{eq21}) to obtain the identities

\begin{eqnarray}
\label{eq26}
\left[ \left[\mathcal{R},\mathcal{P} \right],\mathcal{R}\right]=0
\\ \nonumber
\left[ \left[\mathcal{R},\mathcal{P} \right],\mathcal{P}\right]=0
\end{eqnarray}

\noindent
from which we may obtain

\begin{eqnarray}
\label{eq27}
\left[\mathcal{R},\mathcal{P} \right]
f(\mathcal{R},\mathcal{P}) = &&\sum\limits_{m,n = 0}^\infty  \frac{1}{n!m!}\frac{\partial ^{n + m} f(\mathbf{r},\mathbf{p})}
{\partial \mathbf{p}^n \partial \mathbf{p}^m } |_{\mathbf{r},\mathbf{p} = \mathbf{0}} \\
\nonumber &&
\left[\mathcal{R},\mathcal{P} \right]\textsf{S}\{\mathcal{R}^n \mathcal{P}^m \}\\
\nonumber &&
= \sum\limits_{m,n = 0}^\infty  \frac{1}{n!m!}\frac{\partial ^{n + m} f(\mathbf{r},\mathbf{p})}
{\partial \mathbf{p}^n \partial \mathbf{p}^m } |_{\mathbf{r},\mathbf{p} = \mathbf{0}} \\
\nonumber &&
\textsf{S}\{\mathcal{R}^n \mathcal{P}^m \}\left[\mathcal{R},\mathcal{P} \right]
\end{eqnarray}

\noindent
Therefore, we can readily write down that
\begin{equation}
\label{eq28}
\left[\left[\mathcal{R},\mathcal{P} \right],f(\mathcal{R},\mathcal{P})\right]=\left[f(\mathcal{R},\mathcal{P}),\left[\mathcal{R},\mathcal{P} \right]\right]
\end{equation}

\noindent
which in turn results in $\left[\mathcal{R},\mathcal{P} \right]={\rm const}$. Therefore, comparing to (\ref{eq20}), we obtain the same generalized form of Poisson bracket, that is $\left[\mathcal{R},\mathcal{P} \right]=i\hbar{\rm I}$.

\end{proof}

Similarly, applying the same calculations to the conjugate evolution equations (\ref{eq3}) and (\ref{eq13}) we obtain a comparable commutator between energy $\mathcal{H}$ and time $\mathcal{T}$ operators as

\begin{equation}
\label{eq29}
\left[\mathcal{H},\mathcal{T} \right]=i\hbar
\end{equation}

\subsection{Uncertainty Relationships}

Following the standard procedure, once the commutators between two operators such as $\mathcal{A}$ and $\mathcal{B}$ is known, then it can be shown that \cite{Sakurai,Schleich,Khorasani}

\begin{equation}
\label{eq30}
\Delta a \Delta b \geq \frac{1}{2}|\langle\left[\mathcal{A},\mathcal{B}\right]\rangle|
\end{equation}

\noindent
where the standard deviations, or uncertainties $\Delta a$ and $\Delta b$ are given by
\begin{eqnarray}
\label{eq31}
\Delta a=\sqrt{\langle\mathbb{A}^2\rangle-\langle\mathbb{A}\rangle^2} \\ \nonumber
\Delta b=\sqrt{\langle\mathbb{B}^2\rangle-\langle\mathbb{B}\rangle^2}
\end{eqnarray}

\noindent
Here, expectation values are taken with respect to a given state function such as $\psi(\mathbf{r})=\langle\mathbf{r}|\psi\rangle$ or $\chi(\mathbf{p})=\langle\mathbf{p}|\chi\rangle$ as

\begin{eqnarray}
\label{eq32}
\langle\mathbb{A}\rangle&&=\iiint \psi^{*}(\mathbf{r})\mathcal{A}\psi(\mathbf{r})d^{3}r
\\ \nonumber
&&=\iiint \chi^{*}(\mathbf{p})\mathcal{A}\chi(\mathbf{p})d^{3}p
\end{eqnarray}

\noindent
Anyhow, we readily recover the famous uncertainty relationships

\begin{eqnarray}
\label{eq33}
\Delta \mathbf{r}\cdot\Delta\mathbf{p}\geq\frac{3}{2}\hbar
\\ \nonumber
\Delta E\cdot\Delta t\geq\frac{1}{2}\hbar
\end{eqnarray}

\noindent
where we have taken note of the fact that

\begin{equation}
\label{eq34}
\Delta \mathbf{r}\cdot\Delta\mathbf{p}=\Delta x\Delta p_x+\Delta y\Delta p_y+\Delta z\Delta p_z
\end{equation}

\section{Conclusion}

In this article, we have demonstrated that the non-commutative algebra of quantum mechanics, and in particular, the Poisson bracket can be indeed derived, instead of being postulated. However, this would rely on modifying the basic postulates of quantum mechanics and shifting to more basic concepts starting from evolution equations and conjugate variables.

\appendix
\section{Operators in Ket and Function Spaces}

Any operator acting on any member of the ket space such as $\mathbb{A}|\psi\rangle=|\phi\rangle$ relates to a dual operator such as $\mathcal{A}$ acting on the functions $\psi(a)$ or obeys a similar relationship such as $\mathcal{A}\psi(a)=\phi(a)$, where $\psi(a)=\langle a|\psi\rangle$ and $\phi(a)=\langle a|\phi\rangle$. This is called a projection of ket space unto $a$ space of functions. Although normally $\mathbb{A}$ and $\mathcal{A}$ refer to the same physical quantities they are not mathematically the same. While $\mathbb{A}$ operates in the ket sapce, $\mathcal{A}$ operates in the function space. These two however are transformable using the second quantization as \cite{Khorasani}

\begin{equation}
\label{A1}
\mathbb{A} = \iiint {\hat \psi ^\dag  \left( {\mathbf{r}} \right)\mathcal{A}\hat \psi \left( {\mathbf{r}} \right)d^3 r}
\end{equation}

\noindent
where field operators are given as

\begin{eqnarray}
\label{A2}
\hat \psi \left( {\mathbf{r}} \right) = \sum\limits_{\left( n \right)} {\hat b_{\left( n \right)} \varphi _{\left( n \right)} \left( {\mathbf{r}} \right)}\\ \nonumber
\hat \psi ^\dag  \left( {\mathbf{r}} \right) = \sum\limits_{\left( n \right)} {\hat b_{\left( n \right)} ^\dag  \varphi _{\left( n \right)} ^* \left( {\mathbf{r}} \right)}
\end{eqnarray}

\noindent
Here, $\hat b_{\left( n \right)}$ and $\hat b_{\left( n \right)} ^\dag$ are respectively the annihilation and creation operators at the $(n)$th state. Also
$\varphi _{(n)}({\mathbf{r}})=\langle \mathbf{r}|(n)\rangle$, where $|(n)\rangle$ is an eigenket of the system. Obviously, this is a projection unto the three-dimensional position $\mathbf{r}$, while any other projection such as unto the three-dimensional momentum $\mathbf{p}$ could have been also used.

\section{Time Operator and Time Eigenstates}

For a given physical system under consideration, we may assume a time operator in the ket space here denoted by $\mathbb{T}$. In general, $\mathbb{T}$ can be a function of Hamiltonian, coordinates and momentum as $\mathbb{T=T(H,R,P)}$. Hence, we may recast the corresponding evolution equation in the ket space as

\begin{equation}
\label{B1}
\mathbb{T(H,R,P)}|\chi(e)\rangle=-i\hbar\frac{\partial}{\partial e}|\chi(e)\rangle
\end{equation}

\noindent
where $|\chi(e)\rangle$ is the state ket of the system. If the time operator is independent of the total energy, or Hamiltonian, then

\begin{equation}
\label{B2}
\mathbb{T(R,P)}|\chi_{t}\rangle=t|\chi_t\rangle
\end{equation}

\noindent
in which $t$ is the time eigenvalue, and $|\chi_t\rangle$ is the time eigenstate. We may notice that $|\chi_t\rangle$ and $|\chi(e)\rangle$ are related as

\begin{equation}
\label{B3}
|\chi(e)\rangle=\exp(\frac{i}{\hbar}te)|\chi_{t}\rangle
\end{equation}

To illustrate a basic example, consider the operator

\begin{equation}
\label{B4}
\mathbb{K}=\frac{1}{2\hbar^2\Omega^2}\mathbb{H}^2+\frac{\Omega^2}{2}\mathbb{T}^2
\end{equation}

\noindent
Here, $\Omega$ is a real-valued non-zero constant. Then, we may define the non-Hermitian ladder operators

\begin{eqnarray}
\label{B5A}
\hat{b}=\frac{1}{\sqrt{2\hbar}}(\frac{1}{\sqrt{\hbar}\Omega}\mathbb{H}+i\sqrt{\hbar}\Omega\mathbb{T})
\\
\label{B5B}
\hat{b}^\dagger=\frac{1}{\sqrt{2\hbar}}(\frac{1}{\sqrt{\hbar}\Omega}\mathbb{H}-i\sqrt{\hbar}\Omega\mathbb{T})
\end{eqnarray}

\noindent
Then, it would be easy to verify the identity $[\hat{b},\hat{b}^\dagger]=1$, and thus

\begin{equation}
\label{B6}
\mathbb{K}=\hat{b}^\dagger\hat{b}+\frac{1}{2}
\end{equation}

\noindent
Evidently, the operators $\mathbb{K}$, $\hat{b}$ and $\hat{b}^\dagger$ obey the algebra

\begin{eqnarray}
\label{B7A}
\mathbb{K}|m\rangle=(m+\frac{1}{2})|m\rangle
\\
\label{B7B}
\hat{b}|m\rangle=\sqrt{m}|m-1\rangle
\\
\label{B7C}
\hat{b}^\dagger|m\rangle=\sqrt{m+1}|m+1\rangle
\end{eqnarray}

\noindent
Then, the time representation in the function space will be given by $\chi_m(e)=\langle e |m\rangle$. Similarly, the energy representation will be given by $\phi_m(t)=\langle t|m\rangle$.


\end{document}